\documentclass[conference]{IEEEtran}
\usepackage{amssymb,amsmath,amsthm}
\usepackage{caption}
\usepackage{subcaption}
\usepackage{graphicx}
\usepackage{mdwmath}
\usepackage{mdwtab}
\usepackage{soul}
\usepackage{cite}
 \usepackage{multirow}
 \usepackage{color} 
\usepackage{setspace}
\usepackage{esvect}
\usepackage{amsmath}

 \newtheorem{lem}{Lemma}
 
\begin{document}
%
\title{Optical Onion Routing}
\author{\IEEEauthorblockN{Anna Engelmann and Admela Jukan}
\IEEEauthorblockA{Technische Universit\"at Carolo-Wilhelmina zu
Braunschweig, Germany\\
Email: \{a.engelmann, a.jukan\}@tu-bs.de}
}
\maketitle

\begin{abstract}
As more and more data is transmitted in the configurable optical layer, -- whereby all optical switches forward packets without electronic layers involved, we envision privacy as the intrinsic property of future optical networks. In this paper, we propose Optical Onion Routing (OOR) routing and forwarding techniques, inspired by the onion routing in the Internet layer, - the best known realization of anonymous communication today, -- but designed with specific features innate to optical networks. We propose to design the optical anonymization network system with a new optical anonymization node architecture, including the optical components and their electronic counterparts to realize layered encryption. We propose modification to the  secret key generation using Linear Feedback Shift Register (LFSR), -- able to utilize different primitive irreducible polynomials, and the usage optical XOR operation as encryption, an important optical technology coming of age. We prove formally that, for the proposed encryption techniques and distribution of secret information, the optical onion network is perfectly private and secure. The paper aims at providing practical foundations for privacy-enhancing optical network technologies. 
\end{abstract}


%
\IEEEpeerreviewmaketitle

\section{Introduction}
\par Communication privacy is important. In IP networks, based on the source and destination IP addresses, an adversary can track interactions and interaction patterns, revealing personal data about the users. Therefore, practical mechanisms have been developed to enhance user privacy via unlinkability and unobservability, in the so-called anonymity networks, or mix nets. One of the most popular anonymity networks today is The Onion Routing (Tor), built as an overlay network among volunteer systems on the Internet. Tor provides anonymous communication between source and destination as well as data integrity. Onion routing is a low-latency application of mix nets, where each message is encrypted to each proxy using public key cryptography, with the resulting layered encryption. Each relay has a public and a private key. The public keys are known by all users and are used to establish communication path. Anonymous communication is possible through traffic tunneling over a chain of randomly selected Tor relays. After the tunnel between a pair of Tor routers is setup, symmetric key cryptography is used to transfer the data. These encryption layers ensure sender unlinkability, whereby the eavesdropper is unable to guess complete path from observed links \cite{Erdin:2015,Nepal:2015}.

\par As more and more data is transmitted in the configurable photonic layer, whereby all optical switches and routers forward packets without electronic layers involved, we envision privacy as the intrinsic property of optical networks also. Just like optical and quantum cryptography has advanced the field of traditional cryptography \cite{Mowla:2016,ChenZeng:2015}, optical network systems designed with secrecy and anonymity features should also be able to provide essential building blocks for privacy in future networks, built to serve free societies. However, in contrast to Tor networks, where privacy and anonymity directly depend on number and dependability of volunteer systems, the privacy features in optical network need to be approached differently: for instance, it is a telecom operator that should be able offer a private optical communication service as a value added feature. For instance, for some client networks an optical network can grant anonymous access to third-party servers in the cloud, whereby the traffic contents and the origin of requests can remain secret for both the attacker as well as the cloud provider. In designing an anonymous optical network akin to Tor, however, several obstacles need to be overcome, since the main features need to be primarily implemented in photonics, i.e., without intervention of electronics, such as encryption, traffic routing, and session key distribution. Also, just as Tor requires compute intensive processing of encryption layers in forwarding routers, high speed processing of optical data would also be required, or consideration of large optical buffers, which is a challenge, and requires practical foundations for privacy-enhancing optical network technologies. 
    \begin{figure*} [t]
  \centerline{\includegraphics[width=0.75\textwidth]{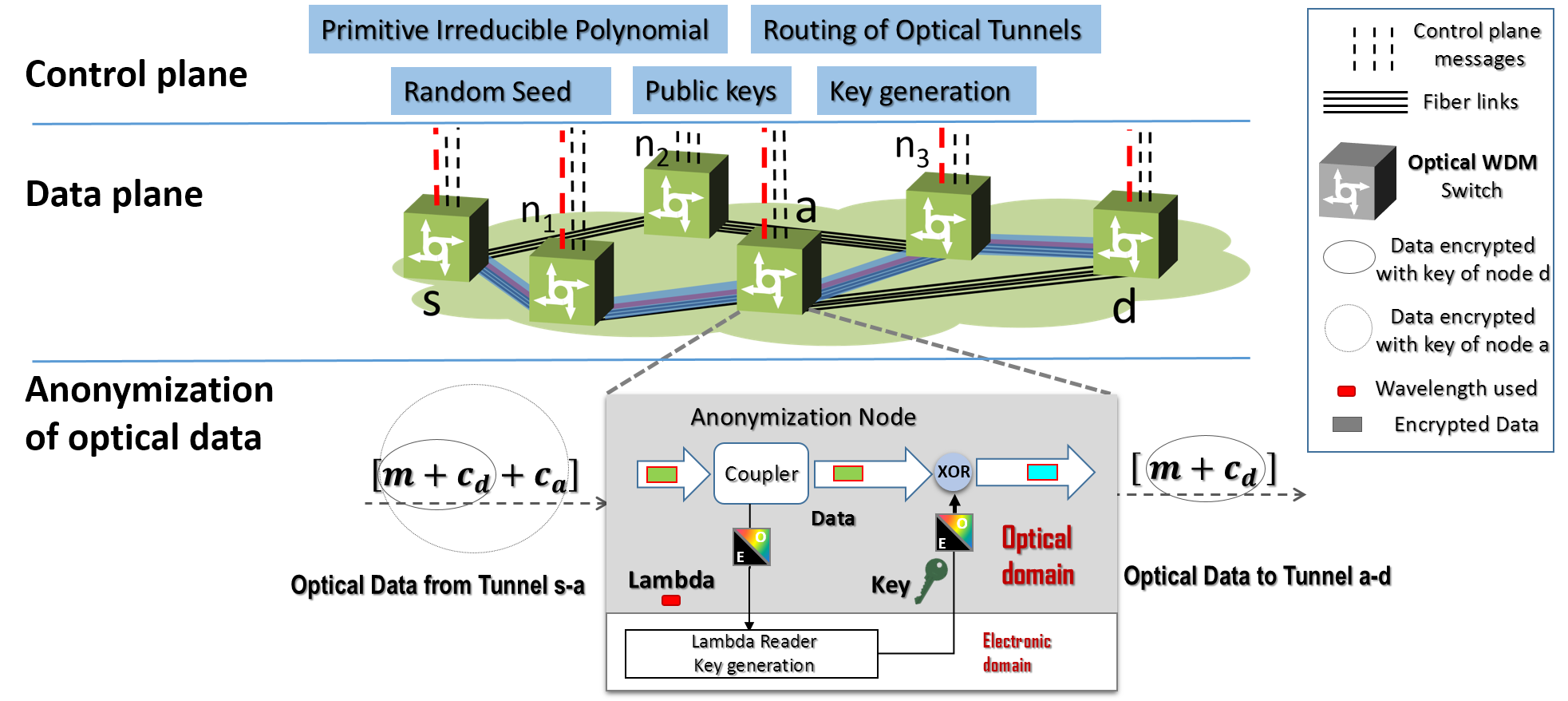}}
  \vspace{-0.2cm}
  \caption{\small Anonymous forwarding through data transformation.}
  \label{transform}\vspace{-0.6cm}
  \end{figure*}

\par In this paper, we propose to treat the well-known privacy constructs of Tor in the optical layers, which we refer to as the Optical Onion Routing (OOR).  To this end, we address two practical issues: the design of all-optical anonymization nodes, and the secrecy and privacy degree achieved. To design an optical anonymizaiton node, we propose to generate a session key with Linear Feedback Shift Register (LFSR), -- a component more commonly used for random number generation, able to utilize different primitive irreducible polynomials of random degree. In addition, we propose to use an optical XOR operation as encryption, an important all-optical technology coming of age. These two components allow to encrypt data in the optical layer at the line speed, thus eliminating the need for large buffers in the node.  To enable optical routing and forwarding, we integrate the anonymization functions in the traditional optical cross connect architecture, with the goal of processing optical data all-optically as much as possible with the current technologies. Finally, we prove formally that for the encryption technique and distribution of secret information proposed, the system can be perfectly private and secrecy-preserving, whereby entropy of the secret data is equal to or less than equivocation observed by the fiber eavesdropper.

\par The rest of the paper is organized as follows. Section II provides principles of the optical anonymity routing proposed. Section III presents the analysis. Section IV shows analytical and simulation results. Section V concludes the paper.

\section{System Model}
\subsection{Anonymous forwarding}
\par Onion routing in the Internet is based on a connection-oriented communication channel, \emph{a circuit}. This is where we start drawing the analogy. We envision optical WDM network as the underlaying infrastructure to setup that circuit, and assume a network of optical nodes and fiber links, where switching, routing and forwarding is all done in the photonic domain. Just like in the Tor, the nodes can act as either regular optical nodes, - with all-optical switching and forwarding functions, or the anonymization nodes. Anonymization nodes are the optical nodes with enhanced functions responsible for processing and forwarding optical signals, such that no correlation can be established between the source and destination by tapping into any link along the way.  As the optical network architecture usually encompass both the data plane, and from the data plane a separated control plane, we assume that the control plane is able to provide information about network topology, available network resources and is able to direct optical data and control the related processing such as encryption. Similar to Tor, only a subset of anonymization nodes in the network is enough to assure secrecy and anonymity. Control plane randomly selects anonymization nodes in the network and available wavelengths and, then, sends control message  to establish optical circuit between source $s$ and destination $d$ on the select wavelengths. Here, control plane does not distribute the actual session keys, but only the routing information for optical circuit setup as well as randomly selected parameters for session key generation. To keep these sensitive control information private and confidential, the control plane encrypts it in layers by applying the public key cryptography, just like in Tor. Fig.~\ref{transform} illustrates the idea of OOR network architecture.

\par The source $s$ is an initiator of private communication, whereby, based on control plane information, anonymization nodes and the corresponding available wavelengths are randomly selected, and made known to the source. After that, the control plane sends, on a select wavelength or separate control channel, the control messages to establish optical circuit between source $s$ and destination $d$, as well as to distribute policies of session key generation to all nodes in that circuit. This is similar to Tor network, where the tunnel is established over randomly selected IP routers. In our example, the path between source and destination consists of two concatenated circuits, one between nodes $s$ and $a$ and the other one between the nodes $a$ and $d$, whereby each circuit contains one forwarding node; the forwarding node is a traditional optical switching and forwarding function, without anonymization. The circuit, i.e., end-to-end wavelength path, is established over arbitrary available links and forwarding nodes on the available wavelengths. Thus, the path available between $s$ and $d$ is randomly selected for setup, so that neither the destination $d$ nor anonymization nodes know the paths selected (which is the essence of Tor). The control message is encrypted with public keys of nodes $a$ and $d$, as in Tor. In contrast to Tor, where data from exit node, i.e., the last anonymization node, to destination is sent without encryption, all nodes involved in anonymous communication in OOR, i.e, $s$, $a$ and $d$, perform anonymization of optical data via encryption.

\par The idea behind onion routing, and its Tor implementation, is to hide the communicating nodes from the eavesdropper of the individual links, as well as the identity of the source from the destination. This is how we envision to do it in the optical layer.  After the optical path (tunnel) is established (via two circuits) the secret data $m$ is ready for transmission towards the anonymization node $a$.  The secret data $m$ is encrypted at the source with a session keys $c_a$ and $c_d$ of $a$ and $d$, respectively. These session keys are generated with the previously mentioned Linear Feedback Shift Register (LFSR), which is its new application as it is a component more commonly used for random number generation. Here, LSFR generates the key based on randomly selected generator polynomials and seeds configured by the control plane. Thus, the source sends an optical stream $[m+c_d+c_{a}]$ to node $a$. \emph{Lambda reader} in node $a$ detects the input port and wavelength allocated, and based on that allocated wavelength the \emph{Key Generation Unit} allocates a suitable session key $c_a$ which is then sent to optical Decryption Unit. Finally, the  payload is decrypted with a key $c_a$ as $m+c_{d}+c_{a}+c_{a}=m+c_{d}$. Next, the optical data stream $[ m+c_d]$ sent by $a$ over sub-tunnel reaches node $d$.  After detecting of input signal on certain wavelength, at $d$, the session key $c_d$ is applied to optical payload as $m+c_d+c_d=m$. Due to data encryption in anonymization nodes, each outgoing optical stream differs from incoming optical stream. When an attacker has access to links of certain switch, it must deanonymize all outgoing data to identify a certain optical stream of interest and to guess its next hop.

\subsection{Discussion on implementation}

\subsubsection{Public Key Cryptography}
To distribute confidential control plane information during optical circuit (tunnel) setup process, the public key cryptography can be applied, similar to Tor. The public key cryptosystems require two keys, i.e., public $K^+$ and private $K^-$. The public keys of all nodes in the network are known. Due to the fact, that the key information must be stored, the public key cryptosystem is implemented in the electronic control plane layer, similar to what is proposed in \cite{Guneysu:2014}. In our architecture, we do not define a specific public key cryptosystem and generally allow all existing public key designs, which can be based on discrete logarithm problem such as Diffie-Hellman, factorization problem such as Rivest-Shamir-Adelmann (RSA) or on square root problem, such as Rabin systems.

\begin{figure} [t!]
\hspace{1mm}
  \centerline{\includegraphics[width=0.38\textwidth]{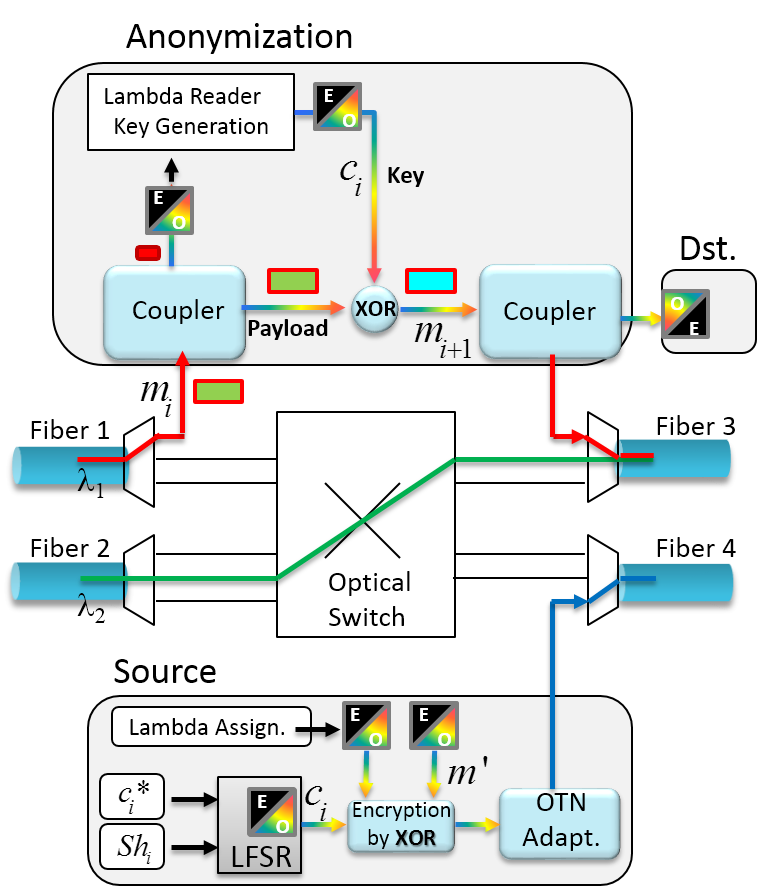}}
  \vspace{-0.4cm}
  \caption{\small Architecture of optical node in OOR network.}
  \label{node}\vspace{-0.6cm}
  \end{figure}

\subsubsection{Exclusive or (XOR) operation}
The XOR operation  utilized in cryptographic systems is usually implemented in software. We propose to implement encryption and decryption with session keys in the optical layer, i.e., data anonymization, with an all-optical XOR gate component \cite{Dimitriadou:2013,Yang:2010}. The XOR operation transforms the incoming data into new outgoing data and, thus, unlinks the communication between source and destination, whereby  each incoming message is mixed, i.e., XOR concatenated, with a session key. Here, ultrafast nonlinear interferometers based on semiconductor optical amplifiers (SOAs) can be used to combine two optical streams, whereby transverse electric (TE) and transverse magnetic (TM) components of a probe pulse can be split and recombine by setting the relative optical delays between them. When the phases experienced by the TE and TM components in the SOA are the same, the resulting signal is '1', or '0' otherwise.

\subsubsection{Linear Feedback Shift Register}
  In each anonymization node, optical data is ananymized by encryption, before it is forwarded to output port. We propose to generate the session key for ananymization by LFSR, a component commonly used as random number generator. The session key generation with LFSR is discussed in \cite{67207, Eljadi:2014}.  Since LFSR of length $n$ bits can generally be easily deducted by observing $2n$ consecutive generated bits, we propose to utilize different generator polynomials of different degrees and randomly selected seeds. This can help us to increase the amount and randomness of possible session keys, with the goal to provide a one-time pad, which is random, and at least as long as the plaintext, and not reused and completely secret \cite{Shannon:1949}. Generally, LFSR can be implemented in hardware or software. In our system, this function is implemented in the optical layer, which necessitates an electrical-to-optical conversion before encryption (XOR). The session key can be pre-calculated during circuit setup process, or generated at line rate. The first variant is suitable for LFSR implementations as we propose, based on \cite{David:2012},  whereby it must be assumed that additional electronic buffer is required to store the pre-calculated keys. In contrast, the second solution must be implemented at line speed, what is a challenge for current optical systems, due to high speed, though it would eliminate the need for buffer. 

\subsection{OOR node architecture}

\par A possible node architecture is illustrated in Fig.~\ref{node}. As it can be seen, the typical WDM node architecture is enhanced to provide functions of anonymization. In that sense, the node can act as a simple forwarding (all-optical switching node), anonymization node, or OOR node for sending (source) or termination (destination). We next describe each of these functionalities and concepts in more detail.

\subsubsection{Source }
The basic function of the source node is to modulate the electronic signals onto optical carriers, along with the flow encryption, the wavelength assignment and/or any other flow adaptation for further transmission over an OTN/WDM network. Here, the optical data generated ($m'$) is encrypted with the dedicated keys from LFSR and by applying optical XOR. The source collects all anonymization keys of the anonymization nodes that are to be used on the wavelength path assigned, e.g., $\vv c=\{c_1, c_2,...,c_d \}$ and these keys are to be utilized by each anonymization node traversed, e.g., $c_1$ by the node $a_{1}$ and $c_d$ by destination $d$; this is a way to anonymize the incoming optical data for the next hop or to decrypt it. The generation of the key $c_i$ is important and it is created in LFSR with polynomial $c_i^*$ from vector $\vv c^*$ and seed $Sh_i$ from $\vv {Sh}$, both randomly selected in the control plane and distributed to each node $a_i$ during the circuit setup. The incoming optical signal $m\rq{}$ is then encrypted with optical XOR by all elements from $\vv c$ as $[m'+c_{1...d}]$, where $c_{1...d}=c_1+c_2+...+c_d$. The encrypted optical data $[m'+c_{1...d}]$ is finally sent to the predefined optical circuit (on Fiber 4).

\subsubsection{Anonymization} 
\par Each anonymization node performs data anonymization/decryption before forwarding. Here, the  incoming optical flows $m_i=[m'+c_{i...d}]$ from optical circuit on wavelength $\lambda_1$ (Fiber 1) is detected, and the information about  this wavelength is sent to Lambda Reader and Key Generation Unit for matching. As a result, the corresponding anonymization key $c_i$ is forwarded to optical XOR gate. The session key was also here generated by LFSR by utilizing generator polynomial $c_i^*$ and corresponding seed $Sh_i$, just like in the source node. The session key is simply converted to optical signal, before it is XOR-concatenated with data, as follows $[m'+c_{i...d}+c_i]=[m'+c_{i+1...d}]$. For simplicity, if the data is to be further sent towards the next hop, we assume that the same wavelength is utilized, respecting the wavelength continuity constraint. Otherwise, the signal can also be retransmitted (converted) to another wavelength, which would make it more complex. 

\subsubsection{Destination} 
When data reached its destination $d$, it is processed just like if destination were an ananymization node. The  received optical payload  $[m'+c_d]$ is decrypted with key $c_d$ into $[m']$, converted into the electronic signal at the destination.  

\section{Modeling and analysis}
\subsection{Routing and Treat Model in OOR}
\subsubsection {Routing model}
\par We assume that optical circuits in form of wavelength-continuous optical paths are setup in the random fashion over a randomly selected wavelength, whereby a network provides at most $\mathcal N$ optical paths between source $s$ and destination $d$ over all wavelengths and fibers, which for the sake of modeling we collect in set $\Psi$. Generally, only $N$ out of $\mathcal N$, $\mathcal N\geq N$, paths are available, while at least one wavelength paths among them is randomly selected for transmission. All $\mathcal N$ existing optical paths are arranged in the sorted vector $\vv{P}=\begin{pmatrix}
 P_0 \
P_1 \
 ... \
P_{\mathcal N-2} \
 P_{\mathcal N-1} \
\end{pmatrix}$
with related probabilities, that an individual path $\mathcal P_l$ is available.  We denote a fiber link as $e_{v'v''}$ and a wavelength link on $\lambda_x$ connecting two nodes, $v'$ and $v''$, as a wavelength link as $e_{v'v''}(\lambda_x)\in e_{v'v''}$, respectively. The capacity of a fiber link is measured in number of wavelengths. Thus, each edge $e_{v'v''}$ provides $c_{e_{v'v''}}$ parallel wavelength links between nodes $v'$ and $v''$. Each path $\mathcal P_l $ between $s$ and $d$ consists of $\theta_{l}+1$ links $e_{lk}\in \mathcal P_l$, $1\leq k\leq \theta_{l}+1$, and of $\theta_{l}$ intermediate nodes $v_{lq}\in\mathcal P_l$, $1\leq q\leq \theta_{l}$, while $\eta$ out of $ \theta_{l}$ nodes are randomly selected as anonymization nodes $a_{i}$, $1\leq i\leq\eta$. 
\par Let us now assume that there is a collection $\mathcal A$, which contains $\mathfrak a= C(|\Phi|,\gamma):=\binom{|\Phi|}{\gamma}$ path sets $\mathcal A_{\alpha}$, $1\leq\alpha\leq\mathfrak a$, while $\Phi$ can be a collection of all $\mathcal N$ existing paths, i.e., $\Psi$, or of $N$ available paths, i.e., $|\Phi|=\mathcal N$ or $|\Phi|=N$, and $\gamma$ can be a number of available paths $N$ or the number of required for transmission paths, i.e., $\gamma:=N$ or $\gamma:=1$. In contrast, set  $B_{\alpha}=\Phi\backslash A_{\alpha}=\{\mathcal P_{l} |\mathcal P_{l} \notin A_{\alpha} \}$ from collection $\mathcal B$ is the $\alpha^{th}$ set of remaining $|\Phi|-\gamma$ elements, which are not in the $\alpha^{th}$ combination $A_{\alpha}$. Thus, the probability $P''(\alpha,\gamma, \Phi)$, that $\gamma $ paths are in set $A_{\alpha}$ and not in set $B_{\alpha}$, is defined as
\begin{equation}\label{PrPathComb}
P''(\alpha, \gamma, \Phi)=\prod_{i=1,\atop \mathcal P_{l_i}\in A_{\alpha}}^{\gamma}P_{l_i}(\alpha)\prod_{t=1,\atop \mathcal P_{l_t}\in B_{\alpha}}^{\mathcal |\Phi|-\gamma}(1-P_{l_t}(\alpha))
\end{equation}
, where $P_{l_i}(\alpha)$ and $P_{l_t}(\alpha)$ are probabilities of path $\mathcal P_{l}$, $l=1,2,...,\mathcal N$, collected in $A_{\alpha}$ and $B_{\alpha}$ and indexes $i$ and $t$ are the sequence numbers of paths in $A_{\alpha}$ and $B_{\alpha}$, respectively.
\par As a result, the network provides $\Omega=j$ wavelength paths out of $\mathcal N$ paths with probability $\hat P(\Omega=j,\Psi)$ defined as 
\begin{equation}\label{PrNPath}
\hat P(\Omega=j,\Psi)=\sum_{\alpha=1}^{\binom{\mathcal N}{j}}P''(\alpha, j, \Psi)
\end{equation}
, where the $\alpha^{th}$ set from the collection $\mathcal A$ contains one path combination out of $\binom{\mathcal N}{j}$ combinations of $\Omega=j$, $0\leq j\leq\mathcal N$, available paths with related probabilities from vector $\vv P$.   

\par In case of $N<1$ (no path is available), the transmission request will be blocked with probability $P_B$, i.e.,
\begin{equation}\label{RequestBlocking}
P_B=\hat P(\Omega=0, \Psi)
\end{equation}
\par Since we assume that all $N$ paths have the same probability $1/N$ to be selected for transmission, the probability, that any path $\mathcal P_l$ collected in $\mathcal A_{\alpha}$ is available and utilized, is 
\begin{equation}\label{PrComb}
P(\alpha)=\tfrac{P''(\alpha, 1, \Psi)(1-P_B)}{\hat P(\Omega=1, \Psi)}
\end{equation}

\subsubsection{Threat model}
The treat model assumes that an attacker can eavesdrop select links in the network, and guess the source and destination nodes, as well as the data transmitted. To model this, let us define a set  $\mathfrak W$, containing all possible wiretap edges, while at most $\mathbf w = | \mathfrak W| $ edges can be attacked simultaneously. Since optical receivers are broadband, we assume that an attacker is always able to access all $c_{e_{v'v''}}$ wavelengths on a fiber link $e_{v'v''}$. 
\par Let us assume the worst type of attack in the network, where any link in the network can be eavesdropped with a probability $\phi$. In other words, the set of  fiber links attacked $\mathfrak W$ and its size $\mathbf w$ are variable, while each link can belong to set $\mathfrak W$ with probability $\phi$. Here, each wavelength path can be wiretapped with probability $P^w(\mathcal P_l)$ defined by Eq.~\eqref{PrWPath} as a probability that at least one wiretap link utilized by path $\mathcal P_l$.
\begin{equation}\label{PrWPath}
P^w(\mathcal P_l)=\phi\sum_{i=0}^{\theta_l}(1-\phi)^i=1-(1-\phi)^{\theta_l+1}
\end{equation}
, where $\phi$ has the same value for all links in the network. As a result, the probability, that a wiretap path is utilized for transmission, is defined by using of Eqs.~\eqref{PrComb} and ~\eqref{PrWPath} as 
\begin{equation}\label{NumWPathsA}
P^{\phi}_w=\sum_{\alpha=1}^{\mathcal N}P^w(\mathcal P_l)P(\alpha), \forall \mathcal P_{l}\in\mathcal A_{\alpha}
\end{equation}

\subsection{Analysis of data anonymization}
The secret data $m'$ of length $L_{m'}$ bits is sent over OOR network passing through $\eta$, $0\leq\eta\leq\eta_{max}$ anonymization nodes, whereby $\eta_{max}$ is the maximal number of anonymization nodes can be utilized along optical tunnel. When an attacker gains access to encrypted optical data $m$ with probability $P^{\phi}_w$ as discussed previously, it has to decrypt $m$ along all its anonymization keys to reveal the secret data $m\rq{}$.
\begin{lem}\label{secrM}
The OOR system is perfectly secure, whereby an attacker is not able to recover the secret data $m\rq{}$ sent over randomly selected wavelength path.
\end{lem}

\begin{proof} A secret data $m'$ of length $L_{m'}$ bits is generally an arbitrary bit sequence out of all $2^{L_{m'}}$ possible, while the entropy of the plain text is $H(m')=L_{m'}$. $m'$ is encrypted by all $\eta+1$ secret keys of all anonymization node and of the destination.  Thus, there are $(\eta+1)! \cdot \binom{2^{L_{m'}}-2} {\eta+1}\cdot 2^{L_{m'}}$ possible combinations of $m'$ and $\eta+1$ secret keys, while each combination always contains $\eta+2$ different elements out of $2^{L_{m'}}$, whereby only $m\rq{}$ can contain zero element. Thus, the entropy of encrypted data $m'$ is defined as follow
\begin{equation}\label{Hm}
\scalebox{0.95}{\begin{minipage}{2\columnwidth}
$H_e(m')=log\left((\eta+1)! \binom{2^{L_{m'}}-2} {\eta+1}\cdot 2^{L_{m'}}\right)\overset{!}{\geq}L_{m'}=:H(m')$
\end{minipage}}
\end{equation}

\par An attacker does not have any knowledge about the number of selected anonymization nodes $\eta$ or the number of already passed anonymization nodes on the wavelength path and, thus, has to check all $\eta_{max}+1$ possible variants of the same, where $m'$ can be encrypted by one to $\eta_{max}+1$ secret keys. Thus, the equivocation $H(m'|m)$ observed by an attacker is
\begin{equation}\label{HmA}
\scalebox{0.9}{\begin{minipage}{2\columnwidth}
$H(m'|m)=\sum\limits^{\eta_{max}}_{i=0}log\left((\eta_{max}+1-i)! \cdot \binom{2^{L_{m'}}-2} {\eta_{max}+1-i}\cdot 2^{L_{m'}}\right)$
\end{minipage}}
\end{equation}
However, $H(m'|m)\geq H(m')$, thus any $m'$ can be transmitted perfectly secret.
\end{proof}

\par To provide data privacy and anonymity, the proposed OOR utilizes different functional components such as public key cryptography, encoding by XOR and key generation with LFSR on control and data plane. Next, we analyze information-theoretically the resulting privacy and anonymity degree as a function of components utilized. 
\subsubsection{Public Key Cryptography}
The public key cryptosystems require two keys, i.e., public $K^+$ and private $K^-$. The message $m$ sent to node $v_j$ is encrypted by public key $K^+_j$ of $v_j$ as $m_e=K^+_j(m)$. The destination $v_j$ can decrypt received message $m_e$ by applying the private key $K^-_j$ as $K^-_j(m_e)=K^-_j(K^+_j(m))=m$. For high level of data secrecy, we restrict the policies for selecting of key and plain text sizes as $H(m)\leq H(K^+)$, where $H(m)=L_m$ and $H(K^+)=L_K$ are entropies of secret message of length $L_m$ bits and public key of length $L_K$ bits, respectively, i.e., $L_m\leq L_K$. That ensures that an eavesdropper is not able to break the utilized cryptosystem by obtaining the encrypted data $m_e$.

\subsubsection{XOR operation}
We assume that incoming date $m_{v_i}$ of length $L_{m_{v_i}}$ in node $v_i$ is mixed, i.e., XOR concatenated, with a secret key $c_i$ of length $L_{c_i}$ so that an attacker can not recognize $m_{v_i}$ and its next hop node $v_j$. The outgoing data $m_{v_j}$ is defined as $m_{v_j}=m_{v_i}+c_i=\{\forall m_{{v_i}_p}\in m_{v_i} \land \forall c_{{i}_p}\in c_i| (\neg {m_{{v_i}_p}}\land c_{{i}_p})\lor( m_{{v_i}_p}\land \neg c_{{i}_p})\}$, where $m_{{v_i}_p}$ and $c_{{i}_p}$ are the $p^{th}$ bits, $1\leq p\leq L_{m_{v_i}}$, and, $1\leq p\leq L_{c_{i}}$, within $m_{v_i}$ and $c_i$, respectively. Without loss of generality, any secret data $m'$ is XOR encrypted into data $m$ of the same length $L_{m}=L_{m'}$. 

\subsubsection{LFSR}
Generally, keys generated with LFSR do not provide a strong cryptographic security, whereby an attacker is able to gain the generator polynomial of degree $g$, if it receives at least $2g$ consecutive plain text bits generated by LFSR. To this end, we propose to generate session key $c_i$ for data anonymization directly in each anonymization node $a_i$, whereby a primitive irreducible polynomial $c^*_i$ of degree $g$ and seed $Sh_i$ as a start point are randomly selected by source for each utilized anonymization node $a_i$ and secretly distributed with public key cryptography. The source randomly selects one out of $C_g=\varphi(2^g-1)/g$ primitive polynomials of arbitrary degree $g$, $g_{min}\leq g\leq g_{max}$, where $\varphi(x)=x(1-1/p_1)(1-1/p_2)...(1-1/p_k)$ is Euler function, while $p_1...p_k$ are the prim numbers. The minimal degree $g_{min}$ is defined so that the maximal key length generated by LFSR is larger than data $m\rq{}$ encrypted by this key, i.e., $L_{m\rq{}}<2^{g_{min}}-1$.

\par Due to public key cryptography used to distribute control messages during setup process of the optical circuit, all the data is assumed to be perfectly secret. In other words, each control message $m_c$ of length $L_c$ encrypted with a public key $K^+_i$ of node $v_i$ can not be recovered by an attacker, unless the control message $m_c$ out of all $2^{L_c}$ possible messages is guessed, i.e., $H(m_c)=L_c=H(m_c|K^+_i(m_c))\leq H(K^-_i|K^+_i(m_c))=H(K^-_i)$. As a result,  the routing information for circuit provisioining, the bit sequences $\vv c^*$, i.e., randomly selected primitive polynomials for session key generation, and random selected  seeds $\vv Sh$ are perfectly secret, which can not be recovered by an external attacker.

\par  Since the data from source to destination  is anonymized in each anonymization node along optical tunnel, an attacker can only discover $s$ and $d$ by accessing all $\eta+1$ wavelength segments between the anonymization nodes, as well as all incoming links of $s$ and outgoing links of $d$ to ensure that they are not the forwarding nodes of optical data attacked. 


\begin{lem}\label{unlink}
The proposed OOR ensures privacy and secrecy between any $s-d$ pair from an arbitrary extern attacker, if 
\begin{equation}\label{hlen}
\sum_{g=g_{min}}^{g_{max}}log\left( \tfrac{\varphi(2^{g}-1)(2^{g}-1)}{g}\right)\geq L_{m}\leq 2^{g_{min}}-1
\end{equation}
\end{lem}
\begin{proof}
 Let us assume that attacker has access to all links along the path. In this case, the attacker needs to deanonymize the optical data sent by each node $a_{i}$ to the next anonymization node $a_{i+1}$. Due to the fact that the polynomial $c^*_i$ of length $g+1$ bits, $g_{min}\leq g\leq g_{max}$, and seed $Sh_i$ are chosen randomly and transmitted perfectly secure, the entropy of anonymization key can be defined as $H_1(c_i)=\sum_{g=g_{min}}^{g_{max}}log(C_g(2^{g}-1))$, while source randomly selects one out of $C_g$ existing primitive irreducible polynomials of degree $g$ and a seed out of $2^g-1$ (without zero) possible for each anonymization node. On the other hand, the secret key $c_i$ can be an arbitrary bit sequence out of $2^{L_{m}}$ possible, i.e., $H_2(c_i)=L_{m}$ bits. An attacker can follow the algorithm for generation of $c_i$ and, thus, guesses any $c^*_i$ and $Sh_i$ or directly guesses $c_i$ of length $L_{m}$. In the first case, the equivocation is defined as $H(c_i|K^+_i(m_c))=H(c_i|m)=\sum_{g=g_{min}}^{g_{max}}log(C_g(2^{g}-1))=H_1(c_i)$, while, in the second case, $H(c_i|K^+_i(m_c))=H(c_i|m)=L_{m}=H_2(c_i)$. For an attacker, it is simpler to guess polynomial and seed, if $\sum_{g=g_{min}}^{g_{max}} C_g(2^{g}-1)<2^ {L_{m}}$. Thus, $H_1(c_i)$ must be equal to or larger than $H_2(c_i)$ for a perfect secrecy. Since there are $C_g=\varphi(2^g-1)/g$ primitive polynomials of degree $g$, the condition for perfect secrecy provided by anonymization key can be defined by Eq.~\eqref{hlen}, i.e., an attacker will be not able to deanonymize and to link (trace back) to nodes $s$ and $d$.
\end{proof}

\section{Performance evaluation}
\par We now show theoretical results for proposed private and anonymous OOR network and validate the same by simulations. The analytical results were calculated with Eq.~\eqref{NumWPathsA} as well as with Eqs.~\eqref{Hm} and \eqref{HmA}. Since our model directly depends a steady state wavelength path availability and random path selection, we validate the analysis by using dynamic Monte-Carlo-simulations with $95\%$ of confidence. 

\par We analyze modified optical network topology with 24 nodes and 35 fiber links, each fiber link carrying $10$ wavelengths $\lambda$; each wavelength has the capacity of 10Gb/s. The link directions and available number of wavelengths on each fiber link are defined as $\{1-2,1-3,2-6,2-3,3-4,3-7,4-5,4-10,10-11,11-5,6-16, 6-7,16-17,7-17,7-8,17-18,18-8,18-22,8-9,9-4,9-12,9-19,22-19,22-23,19-20,23-24, 23-20,20-12,20-21,24-21,12-10,12-13,21-15,15-13,13-11\}$ and $\{6,6,3,3,4,5,4,3,4,8,2,1,2,3,3,5,2,3,5,3,1,1,1,2,2,1,1,\\2,1,1,1,2,2,2,4\}$, respectively. Let us consider source node 1 and destination node 5. Here, there are in total $\mathcal N = 12$ different possible wavelength paths over all available wavelength links. All paths are sorted in the ascending order of length in number of hops, and collected in $\vv P$. The path availability for each $\lambda$ decreases with increasing path length, i.e., $\vv P = \{0.9,0.85,0.8,0.75,0.75,0.7,0.65,0.6,0.55,0.55,0.5,0.5\}$. Before transmission,  random wavelength paths between anonymization nodes are established by utilizing available wavelengths. Every node in the network can be used as an anonymization node, and the number of anonymization nodes per path is determined randomly.

\par Fig.~\ref{AnonNodes} shows the normalized equivocation $H(m'|m)$ as a function of amount of number anonymization nodes   $\eta$ used on a path and of maximal number of anonymization nodes $\eta_{max}$. An increase in $\eta_{max}$ increases the system robustness against wiretapping (dashed line), while an attacker have to recover more redundant information, when $\eta<\eta_{max}$ anonymization nodes are utilized. For instance, an attacker must recover $35 H_e(m')$ bits to guess secret data $m'$ in case of $\eta=0$ and $\eta_{max}=9$, while increase in $\eta$ increases entropy $H_e(m')$ as per Eq.~\eqref{Hm} and decreases redundant information in encrypted data $m$ up to $6 H_e(m')$ for $\eta=\eta_{max}=9$. 
\begin{figure} [t]
  \centerline{\includegraphics[width=0.9\columnwidth]{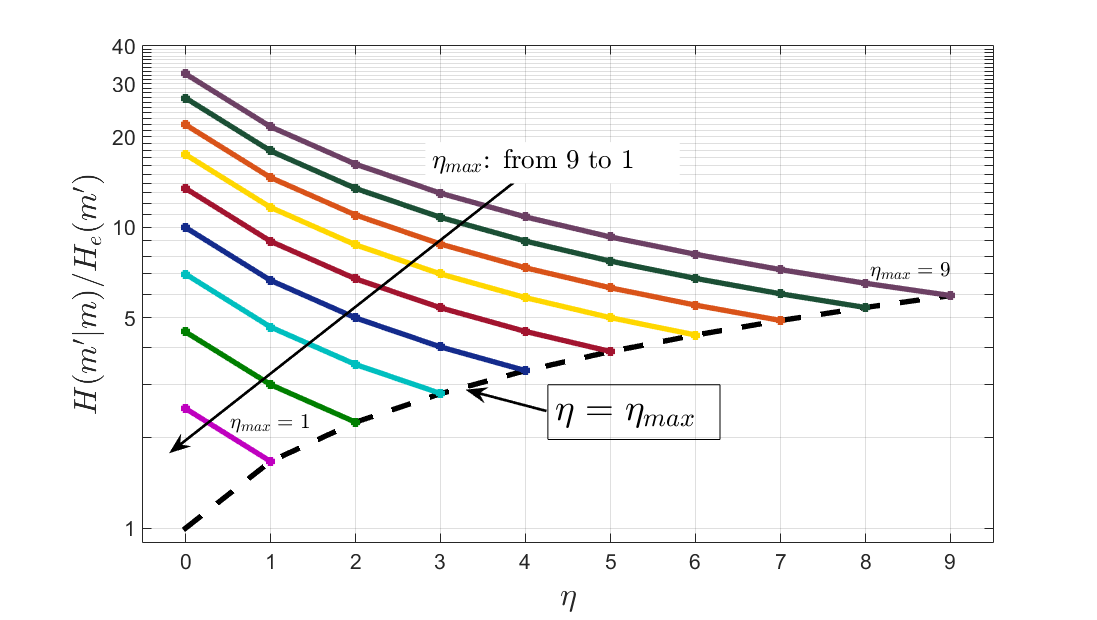}}\vspace{-0.4cm}
  \caption{\small  Normalized equivocation vs. amount of anonymization nodes.}\vspace{-0.4cm}
  \label{AnonNodes}
  \end{figure}
\begin{figure} [t]
  \centerline{\includegraphics[width=0.9\columnwidth]{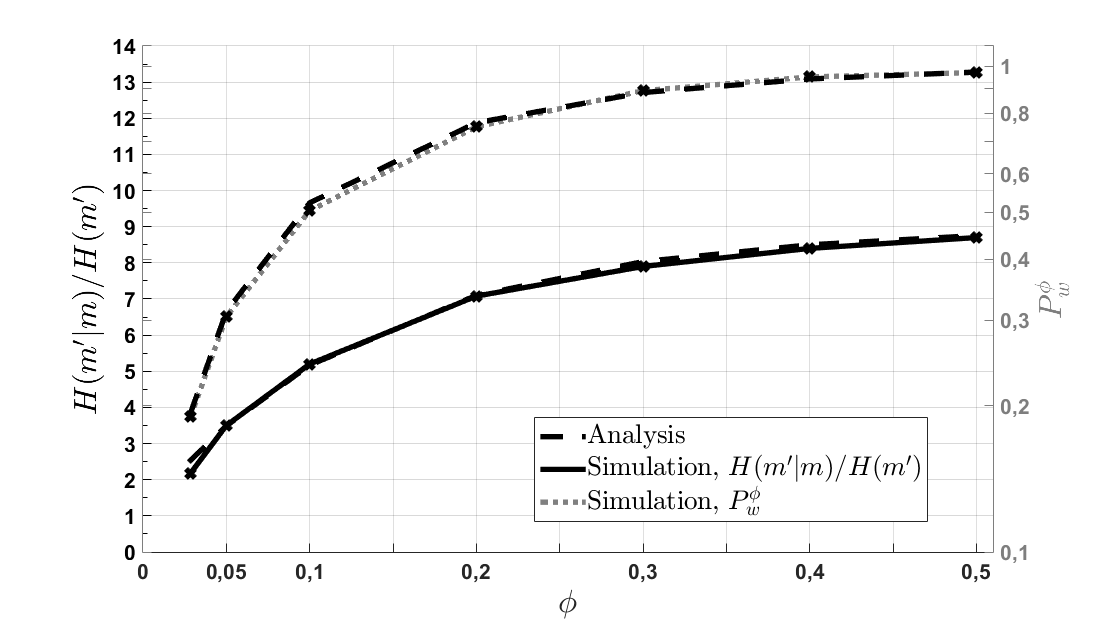}}\vspace{-0.2cm}
  \caption{\small $H(m\rq{}|m)$ and $P^{\phi}_w$ vs. probability for attacked link, $\phi$.}
  \label{PLinksWData}\vspace{-0.5cm}
  \end{figure}
   \begin{figure} [t]
  \centerline{\includegraphics[width=0.9\columnwidth]{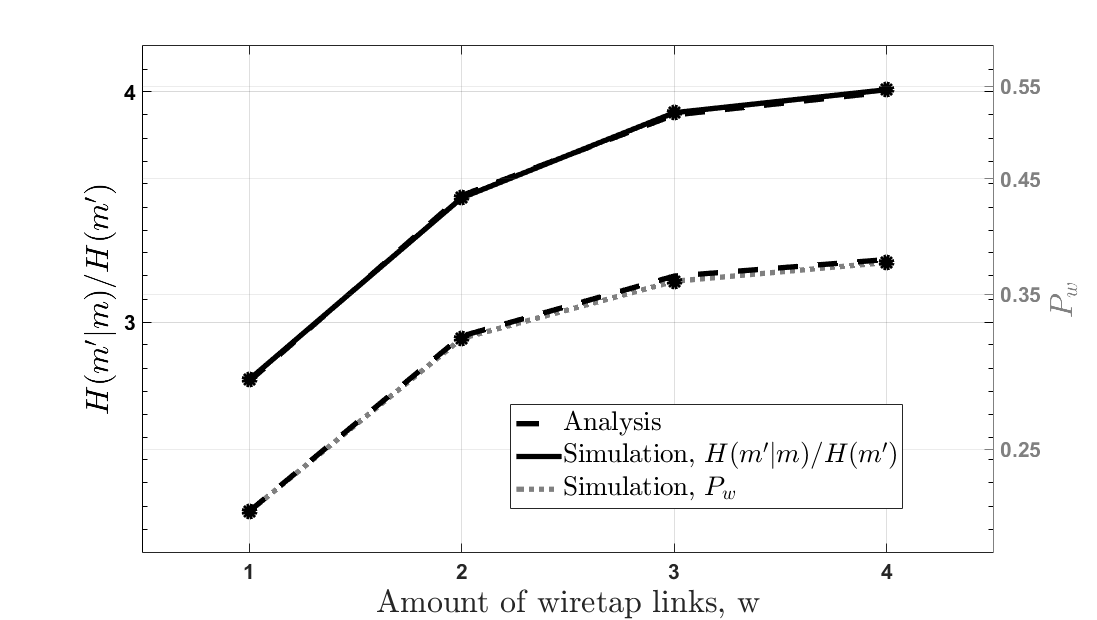}}\vspace{-0.2cm}
  \caption{\small $H(m\rq{}|m)$ and $P_w$ vs. number of wiretap fiber links, i.e., $\mathbf w$.}
  \label{WLinksWData}\vspace{-0.6cm}
  \end{figure}
  \par Next, we assume $\eta_{max}=2$  and evaluate the equivocation $H(m'|m)$ and the probabilities $P_w$ and $P^{\phi}_w$ for successful eavesdropping and correctly recovered data $m\rq{}$. Fig.~\ref{PLinksWData} shows the normalized mean equivocation $H(m'|m)$ and probability for wiretapped transmission path $P^{\phi}_w$, when any link in the network can be eavesdropped with probability $\phi$. The equivocation redundancy and probability for eavesdropped transmission path $P^{\phi}_w$ increase with $\phi$. As a result, an attacker can wiretap almost all paths when probability for wiretap link, $\phi$, is $50\%$, while equivocation redundancy amounts $~9H(m')$, when an attacker tries to decrypt. Next, we consider a special case whereby a maximum of $4$  fiber links in network can be wiretapped either simultaneously or indiviudally, $e_i\in\{ 3-7, 8-9, 17-18, 13-11\}$. Fig.~\ref{WLinksWData} shows the normalized mean equivocation $H(m'|m)$ and probability for wiretapped transmission path $P_w$ as a function of number of fiber links wiretapped at the same time, i.e., $\mathbf w$. An increase in  $\mathbf w$ increases the probability $P_w$ and, thus, the amount of redundant information required to be recovered by attacker, which follows the algorithm to guess $m'$ from eavesdropped optical data $m$. Here, the equivocation increases from around $2.7 H(m')$ to $4 H(m')$ bits with increasing number of wiretap links, i.e., for $\mathbf w=1$ and $\mathbf w=4$, respectively, while the mean amount of wiretapped data, i.e., $P_w$, also increases. 

\section{Conclusion}
We proposed an Optical Onion Routing (OOR) architecture, the mirror of Tor. We designed the network and a new optical anonymization node architecture, including the optical components (XOR) and their electronic counterparts (LFSR) to realize layered encryption. We proved formally and confirmed numerically that such an optical onion network can be perfectly private and secure. The paper aimed at providing practical foundations for privacy-enhancing optical network technologies, and as such is work in progress.

\bibliographystyle{IEEEtran}
\bibliography{codingbibTran}
\end{document}